\documentclass[conference]{IEEEtran}
\IEEEoverridecommandlockouts
\usepackage{cite}
\usepackage{amsmath,amssymb,amsfonts}
\usepackage{algorithmic}
\usepackage{graphicx}
\usepackage{textcomp}
\usepackage{xcolor}
\usepackage{amsthm}

\usepackage{url}
\usepackage{stfloats,bm,bbm,multirow,arydshln,booktabs,subfigure}

\newtheorem{kq}{Key Question}

\newtheorem{game}{Game}
\newtheorem{prob}{Problem}
\newtheorem{defin}{Definition}
\newtheorem{lem}{Lemma}

\newtheorem{thm}{Theorem}
\newtheorem{cor}{Corollary}
\newtheorem{observation}{Observation}

\begin{document}

\title{Information Elicitation from Decentralized Crowd Without Verification}

\author{Kexin Chen, Chao Huang, and Jianwei Huang, \textit{IEEE Fellow}
	\thanks{
		This work is supported by the National Natural Science Foundation of China (Project 62271434), Shenzhen Science and Technology Program (Project JCYJ20210324120011032), Guangdong Basic and Applied Basic Research Foundation (Project 2021B1515120008), Shenzhen Key Lab of Crowd Intelligence Empowered Low-Carbon Energy Network (No. ZDSYS20220606100601002), and the Shenzhen Institute of Artificial Intelligence and Robotics for Society.
		
		Kexin Chen is with Shenzhen Institute of Artificial Intelligence and Robotics for Society (AIRS) and the School of Science and Engineering, The Chinese University of Hong Kong, Shenzhen, Shenzhen 518172, China (e-mail: kexinchen2@link.cuhk.edu.cn). 
		
		Chao Huang is with the Department of Computer Science at the University of California, Davis (email: felixchaohuang@gmail.com).
		
		Jianwei Huang is with School of Science and Engineering, Shenzhen Institute of Artificial Intelligence and Robotics for Society, The Chinese University of Hong Kong, Shenzhen, Shenzhen 518172, China (corresponding author, e-mail: jianweihuang@cuhk.edu.cn).}
}

\maketitle

\begin{abstract}

Information Elicitation Without Verification (IEWV) refers to the problem of eliciting high-accuracy solutions from crowd members when the ground truth is unverifiable.
A high-accuracy team solution (aggregated from members' solutions) requires members' effort exertion, which should be incentivized properly. Previous research on IEWV mainly focused on scenarios where a central entity (e.g., the crowdsourcing platform) provides incentives to motivate crowd members. Still, the proposed designs do not apply to practical situations where no central entity exists. This paper studies the overlooked decentralized IEWV scenario, where crowd members act as both incentive contributors and task solvers. We model the interactions among members with heterogeneous team solution accuracy valuations as a two-stage game, where each member decides her incentive contribution strategy in Stage 1 and her effort exertion strategy in Stage 2. We analyze members' equilibrium behaviors under three incentive allocation mechanisms: Equal Allocation (EA), Output Agreement (OA), and Shapley Value (SV). We show that at an equilibrium under any allocation mechanism, a low-valuation member exerts no more effort than a high-valuation member. Counter-intuitively, a low-valuation member provides incentives to the collaboration while a high-valuation member does not at an equilibrium under SV. This is because a high-valuation member who values the aggregated team solution more needs fewer incentives to exert effort. In addition, when members' valuations are sufficiently heterogeneous, SV leads to team solution accuracy and social welfare no smaller than EA and OA.

\end{abstract}

\section{Introduction} \label{sec: intro}
\subsection{Motivations} \label{subsec: motivation}
% Wisdom of Crowds
The Wisdom of Crowds refers to the fact that groups can often be collectively smarter than individual experts in decision-making \cite{surowiecki2005wisdom}, which has been demonstrated by a series of studies (e.g., \cite{budescu2015identifying, fiechter2021wisdom, kammer2017potential}).
%(e.g., \cite{gurnee1937maze, gordon1924group, treynor1987market}). 
% IEWV
In practical crowd decision applications such as Amazon Mechanical Turk \cite{mturk}, a key challenge is to acquire high-accuracy individual solutions (before they are aggregated into the group's team solution), especially when the ground truth is unavailable. This unverifiable scenario can happen, for example, when the ground truth is related to subjective evaluation (e.g., individuals' preferences for a restaurant \cite{openrice}) or is costly and time-consuming to obtain (e.g., verifying peer assessment results in MOOCs \cite{mooc}). Such problems are known as Information Elicitation Without Verification (IEWV) \cite{huang2020using}. 

% related work 1: centralized setting + reference + key difference
To obtain high-accuracy individual solutions to IEWV problems, members need to exert effort (e.g., consumption of time, energy, and computational resources), which can be costly and should be compensated with proper incentives. There have been increasing studies on IEWV incentive design over the past decade (e.g., \cite{chen2019prior, huang2021strategic, huang2022online,huang2021eliciting}). These studies focused on the scenario where a central entity coordinates the task-solving of crowd members. 
For example, in \cite{huang2021eliciting}, the crowdsourcing platform incentivizes members to exert efforts and truthfully report their solutions.

However, there exist many other IEWV scenarios that do not operate under the orchestration of a central entity. 
For example, a team of students can work together on a course project without a centralized teacher coordination \cite{beneroso2021team}. 
Researchers from different institutes can collaborate on a scientific project and aim to publish high-impact papers based on the aggregation of their research results \cite{wuchty2007increasing,ma2008patent}. In both examples, people team up to finish a task related to their own interests and act in a decentralized manner without an external central entity providing incentives for collaboration.
This paper focuses on this practical yet overlooked decentralized decision scenario.

It is challenging to elicit a high-accuracy team solution (aggregated from heterogeneous members' solutions) to an IEWV problem in a decentralized fashion. First, obtaining a high-accuracy team solution requires members' effort exertion, so members should be properly incentivized. However, in a decentralized setting, there exists no central entity providing incentives to motivate members' effort exertion.
Second, different from the verifiable cases (e.g., prediction markets \cite{witkowski2023incentive}) where members' incentives are computed based on event outcomes, the incentive design is more complex due to the lack of ground truth verification in our setting. Furthermore, no central entity coordinates the design and implementation of incentive allocation.
Hence, it is critical yet very challenging to consider the framework where members act as both incentive contributors and task solvers. Intuitively, more incentives will motivate members' effort exertion and lead to a high-accuracy team solution. However, as each member also acts as an incentive contributor, a high incentive contribution is possible to hurt one's payoff. Therefore, eliciting a high-accuracy team solution without ground truth verification from a decentralized crowd involves highly coupled and complicated member interactions. 

\begin{figure}[h]
	\centerline{\includegraphics[width=.95\linewidth]{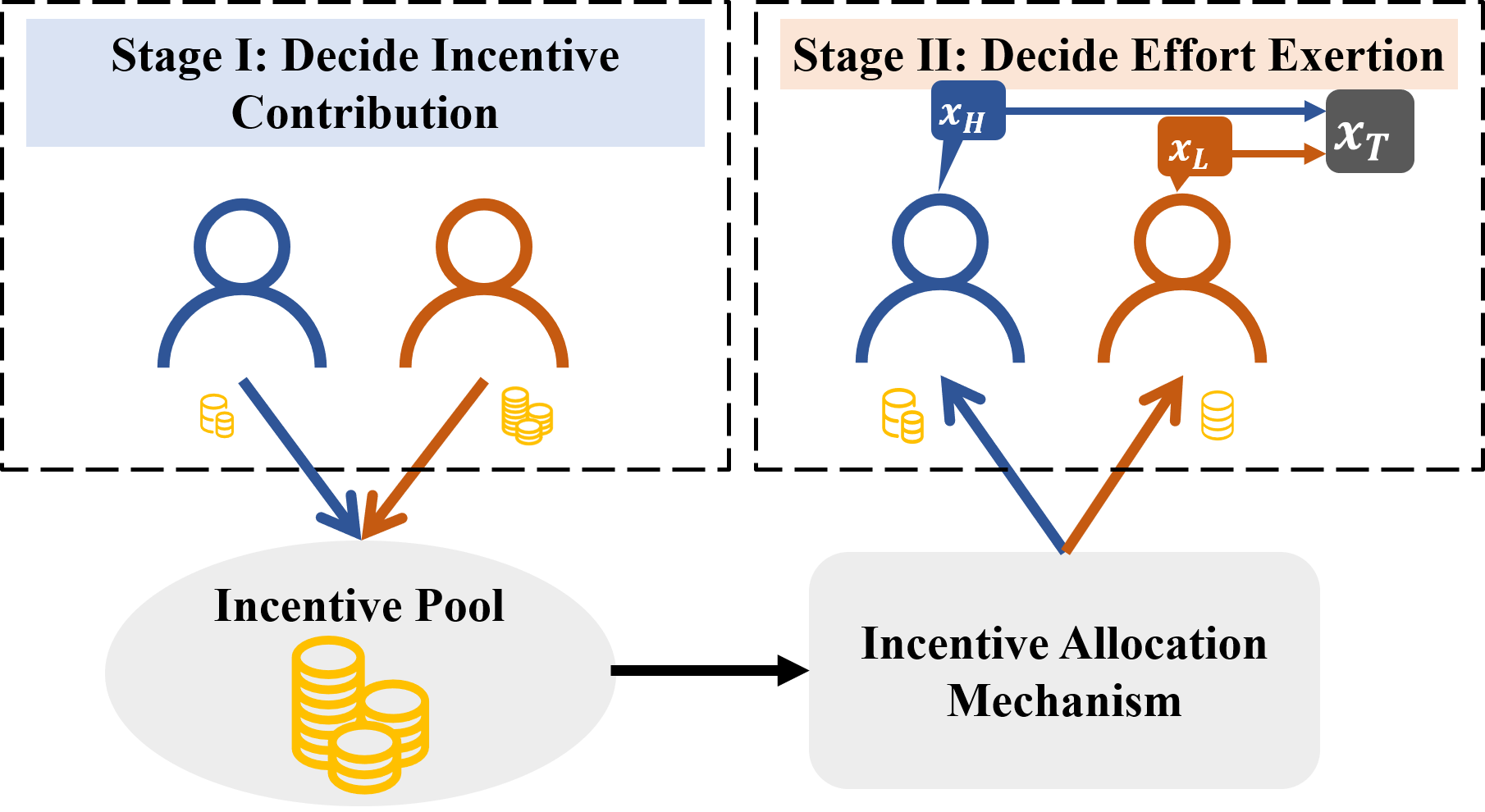}}
	\caption{Two-stage interaction: In Stage \uppercase\expandafter{\romannumeral1}, heterogeneous members choose whether to contribute incentives to the incentive pool. In Stage \uppercase\expandafter{\romannumeral2}, members first decide whether to exert effort in solving the task. Then, members report their solutions ($x_H,x_L$), which will be aggregated into a team solution ($x_T$). Finally, incentives are allocated to the members from the pool according to the predefined incentive allocation mechanism.}
	\label{model ill}
\end{figure}

% model framework
To understand how decentralized members make decisions, we formulate the interaction among members with heterogeneous valuations on the team solution as a two-stage game, where each member decides her incentive contribution strategy in Stage \uppercase\expandafter{\romannumeral1} and her effort exertion strategy in Stage \uppercase\expandafter{\romannumeral2} (see also Fig. \ref{model ill}).
Members' incentive contributions in Stage \uppercase\expandafter{\romannumeral1} will be pooled together and allocated in Stage \uppercase\expandafter{\romannumeral2} based on a predefined incentive allocation mechanism. Their reported solutions will be aggregated into a team solution via the widely adopted majority rule \cite{huang2021eliciting}. The team performance is assessed by: \textit{1) team solution accuracy}: the probability that the aggregated team solution is consistent with the unknown ground truth. \textit{2) social welfare}: members' overall payoffs.

% 3 allocation methods
In this paper, we focus on three incentive allocation mechanisms (which can be implemented without a centralized entity): \textit{1) Equal Allocation (EA)}, which is a benchmark where members equally split the pooled incentives. \textit{2) Output Agreement (OA)}, which is a widely adopted method in centralized IEWV problem \cite{liu2016learning}. Under OA, a member obtains a reward if the reported solution agrees with the majority. \textit{3) Shapley Value (SV)}, which allocates the pooled incentives among members based on their contributions characterized by the Shapley values \cite{shapley201617}. More specifically, each member receives an incentive proportional to her average marginal contribution to the team solution accuracy under SV.

% key questions
In this paper, we aim to address the following key questions.
\begin{kq}
	Given an incentive allocation mechanism (EA, OA, or SV), how should the decentralized members decide their equilibrium incentive contribution and effort exertion strategies?
\end{kq}
\begin{kq}
	Which allocation mechanism achieves the highest team solution accuracy and highest social welfare?
\end{kq}

\subsection{Key Contributions} \label{sec: contribution}
% contributions
The key contributions of this paper are as follows.
\begin{itemize}
	\item \textit{Decentralized IEWV Problem}: 
	To the best of our knowledge, this is the first attempt to study the critical yet under-explored decentralized IEWV scenario. Such a scenario frequently arises in real-world settings where multiple members collaborate to solve a task without the coordination of a central entity, such as completing a joint research project.
	\item \textit{Decentralized Incentive Design}:
	We propose a decentralized incentive framework where members act as both incentive contributors and task solvers. We formulate the highly coupled interactions among members as a two-stage game. The incentives will be first pooled from members and then allocated back to members to encourage effort exertion. Our design provides practical guidance for the decentralized IEWV problem.
	\item \textit{Equilibrium Analysis}:
	We characterize the equilibria under three incentive allocation mechanisms, i.e., EA, OA, and SV. Our analysis reveals that under these mechanisms, the low-valuation member exerts no more effort than the high-valuation member in all equilibria. Counter-intuitively, a low-valuation member provides incentives to the collaboration while a high-valuation member does not at an equilibrium under SV. This is because a high-valuation member who values the team solution accuracy more needs fewer incentives to exert effort.
	\item \textit{Performance Comparison}:
	When members have sufficiently heterogeneous valuations on team solution accuracy, SV leads to team solution accuracy and social welfare no smaller than EA and OA because SV can better incentivize effort exertion. Specifically, a member under SV can always obtain a larger proportion of the incentive amount by exerting effort, while a member under EA or OA may not achieve the same outcome.
\end{itemize}

% overlook of this paper
The remainder of this paper is organized as follows. In Section \ref{sec: model}, we introduce the model. In Section \ref{sec: NE}, we analyze members' equilibrium decisions under three allocation mechanisms. We theoretically and numerically compare the three allocation mechanisms in Section \ref{sec: comparison} and Section \ref{sec: simulation}, respectively. We conclude the paper in Section \ref{sec: conclusion}.

\section{Model} \label{sec: model}
In this paper, we focus on a stylized yet important scenario of two-member cooperation (see also Fig. \ref{model ill}).\footnote{The analysis of the two-member scenario is complex and challenging, as we will demonstrate in this paper. Nonetheless, the insights gained from this analysis will serve as a crucial foundation for understanding more general multi-member scenarios.} This scenario is typical in education and academic contexts, when two members work together on course projects, patent applications, or paper publication tasks.

In Section \ref{subsec: member model}, we introduce members' decisions and payoffs in a decentralized IEWV problem. We formulate the two-stage game in Section \ref{subsec: 2stage} and discuss three incentive allocation mechanisms in Section \ref{subsec: alloc}.

\subsection{Members' Decisions and Payoffs} \label{subsec: member model}
In this subsection, we first define the task and members. Then, we introduce each member's strategies and payoff.

\subsubsection{Task and Members} 
Consider two members $\mathcal{M}=\{H,L\}$ who work on a binary-solution problem,
e.g., two students need to report a team solution to an art appraisal project about whether the painting is \textit{Real} or \textit{Fake}.\footnote{This painting judgment problem is unverifiable due to insufficient evidence of authorship (e.g., \cite{davinci}).}
Define:
\begin{itemize}
	\item $\mathcal{X}=\{1,-1\}$: the solution space, where $1$ represents \textit{Real} and $-1$ represents \textit{Fake}.
	\item $x \in \mathcal{X}$: the true underlying solution (ground truth) unknown to the team members.
	\item $x_m \in \mathcal{X}$: member $m$'s reported solution.
	\item $x_T \in \mathcal{X}$: the aggregated team solution of this task.
\end{itemize}
Each member needs to identify the ground truth $x$ and then report her solution $x_m$ to this task.\footnote{We consider truthful reporting, as we have proved that all the incentive allocation mechanisms in this paper satisfy the truthful property defined in \cite{kong2019information}. We provide the proof in Section \uppercase\expandafter{\romannumeral8} of the online appendix \cite{online}).}
Then, the reported solutions are aggregated into a team solution $x_T$ based on the widely applied \textit{majority rule} as follows \cite{huang2021strategic}.
\begin{equation} \label{eq: mv}
	x_T = \begin{cases}
		1, \qquad & \text{if } x_H + x_L > 0,\\
		-1, & \text{if } x_H + x_L < 0,\\
		1 \text{ or } -1 \text{ with prob. }0.5, &\text{if } x_H + x_L = 0.
	\end{cases}
\end{equation}

Members benefit from generating a high-accuracy team solution, since the final score of a course project positively correlates with the team solution accuracy, i.e., $\Pr(x_T=x)$. 
We use $V_m$ to denote member $m$'s valuation on the team solution accuracy. Without loss of generality, we consider that member $H$ has a high valuation $V_H$ and member $L$ has a low valuation $V_L$, and $V_H > V_L > 0$.

\subsubsection{Member's Incentive Contribution Strategy}
In a decentralized scenario, no central entity provides incentives to members. Without proper incentives, members may be unwilling to exert effort, resulting in a low-accuracy team solution. To address this issue, we consider a setting where each member can contribute incentives to the team to encourage effort exertion. Define $D>0$ as the incentive volume. Let $d_m\in \mathcal{D}_m=\{0,D\}$ denote member $m$'s incentive contribution strategy, where $d_m=0$ means no incentive contribution, and $d_m = D$ means contribution with a fixed volume $D$.
\footnote{Binary-contribution setting is widely adopted in theoretical analysis (e.g., \cite{Chen_2021}), since one can always find an equilibrium in a finite game. We plan to study the more general incentive contribution setting in future work.} We use a vector $\bm{d}=(d_m,\ \forall m\in \mathcal{M})$ to denote members' incentive contribution profile.
The incentives will be gathered in an incentive pool and reallocated to members.

Our model is flexible enough to accommodate different forms of incentives. In the context of a course project, the total amount of points (that depend on the accuracy of the team solution) will be distributed among the team members (hence serving as the incentives). Similarly, in a paper publication (or a patent application), the ranking of the authorship (or inventors) serves as the indication of incentive. 

\subsubsection{Member's Effort Exertion Strategy} \label{subsubsec: eff}
Each member can decide whether to exert effort $e_m \in \mathcal{E}_m=\{0,1\}$ \cite{liu2017sequential} to finish the task, where $e_m = 0$ means that she puts in no effort, and $e_m = 1$ means that she exerts effort with a cost $c$ (e.g., consumption of time, energy, or computational resources). Exerting effort improves the accuracy of a member's solution. Let $q_m (e_m)$ denote member $m$'s solution accuracy, i.e., the probability that the reported solution is correct (i.e., $x_m=x$). More specifically,
\begin{equation} \label{ability}
	q_m(e_m)=
	\begin{cases}
		0.5,\ & \text{if }e_m=0 \text{, with no cost},\\
		a\in(0.5,1], & \text{if }e_m =1 \text{, with a cost }c > 0,
	\end{cases}
\end{equation}
where $a$ is member $m$'s probability of generating a correct solution if she puts in effort.\footnote{We assume that members have homogeneous effort cost and homogeneous ability to generate a correct solution with effort exertion. Although the setting is relatively simple, analyzing members' behaviors is challenging due to members' different valuations and coupled decisions.} 

We further define $\boldsymbol{e} = (e_m,\ \forall m \in \mathcal{M})$, and $P_T(\boldsymbol{e}) = \Pr(x_T=x)$, i.e., the probability that the team solution is consistent with the ground truth. 

\begin{lem}[Team Solution Accuracy] \label{lem: team acc}
	The probability that the team solution is consistent with the ground truth is
	\begin{equation} \label{eq: team acc}
		P_T(\boldsymbol{e}) = \begin{cases}
			\frac{1}{2}, & \text{\textup{if} }e_H+e_L=0,\\
			\frac{2a+1}{4},\quad & \textup{if }e_H+e_L=1,\\
			a, & \textup{if }e_H+e_L=2.\\
		\end{cases}
	\end{equation}
\end{lem}

\begin{proof}[Proof sketch]
	We first calculate the probabilities of all possible reporting results.
	Next, for each reporting profile, we calculate $P_T(\bm{e})$ based on \eqref{eq: mv}.
\end{proof}

We provide the proof for Lemma \ref{lem: team acc} in Section \uppercase\expandafter{\romannumeral1} of the online appendix \cite{online}. The intuition behind Lemma \ref{lem: team acc} is that more members exerting effort improves the team solution accuracy. Therefore, we need to design proper incentives to motivate members' effort exertion.

\subsubsection{Incentive Allocation}
To encourage effort exertion, one needs a proper incentive allocation mechanism.
Let the mapping $\mathcal{A}:\{e_m\}_{m\in \mathcal{M}} \mapsto \{p_m^\text{z}(\boldsymbol{e})\}_{m\in \mathcal{M}}$ denote a general allocation mechanism z (which we will explain in Section \ref{subsec: alloc}), where $p_m^\text{z}(\boldsymbol{e}) \in [0,1]$ is the expected proportion of the total incentives that member $m$ can get from the incentive pool.

\subsubsection{Member's Payoff}
We define member $m$'s payoff as

\begin{align} \label{payoff}
	U_m(\boldsymbol{d,e}) = V_m \cdot P_T(\boldsymbol{e}) + p_m^\text{z}(\boldsymbol{e}) \cdot \sum_{n \in \mathcal{M}} d_n - d_m - c \cdot e_m.
\end{align}

The term $V_m \cdot P_T(\boldsymbol{e})$ captures the utility of member $m$, which is positively correlated with the team solution accuracy and depends on both members' effort exertion decisions. The term $p_m^\text{z}(\boldsymbol{e}) \cdot \sum_{n \in \mathcal{M}} d_n$ captures the expected incentive that member $m$ can get from the incentive pool under the allocation mechanism z, $d_m$ captures member $m$'s contributed incentives, and $c \cdot e_m$ captures the cost for effort exertion.

\subsection{Two-Stage Game Formulation} \label{subsec: 2stage}
% intro + illustration + justify why we can decide the two stages in this sequence
In this subsection, we introduce the two-stage interactions among members. With decentralized information elicitation, members can improve the team performance by contributing incentives in Stage \uppercase\expandafter{\romannumeral1}, since the contribution will motivate members' effort exertion in Stage \uppercase\expandafter{\romannumeral2} under a proper incentive allocation.

\subsubsection{Incentive Contribution in Stage \uppercase\expandafter{\romannumeral1}}
As each member's incentive contribution strategy affects every member's payoff, members play a game defined as follows. 

\begin{game}[Incentive Contribution Game in Stage \uppercase\expandafter{\romannumeral1}] \label{game: incentive}
	The incentive contribution game is a tuple $\Omega = (\mathcal{M,D},\boldsymbol{U})$ defined by:
	\begin{itemize}
		\item Players: The set $\mathcal{M}$ of members.
		\item Strategies: 
		Each member $m$ chooses an incentive contribution strategy $d_m \in \mathcal{D}_m$. The feasible set of all strategy profiles is $\mathcal{D} = \prod_{m\in \mathcal{M}} \mathcal{D}_m$.
		\item Payoffs: The vector $\boldsymbol{U} = (U_m,\ \forall m \in \mathcal{M})$ contains each member's payoff as defined in \eqref{payoff}.
	\end{itemize}
\end{game}

In Game \ref{game: incentive}, given the other member's decision $d_{-m}$, member $m$ solves the following problem to find her best response.

\begin{prob}[Member $m$'s Best Response Problem in Stage \uppercase\expandafter{\romannumeral1}] \label{pr: incentive}
	\begin{equation} \label{br: incentive}
		\begin{aligned}
			\max && U_m(d_m,d_{-m})\\
			\text{s.t.} && d_m \in \mathcal{D}_m.
		\end{aligned}
	\end{equation}
\end{prob}

\subsubsection{Effort Exertion in Stage \uppercase\expandafter{\romannumeral2}}
As a member's effort exertion affects the team solution accuracy and hence the other member's payoff, members play an effort exertion game defined below. 

\begin{game}[Effort Exertion Game in Stage \uppercase\expandafter{\romannumeral2}] \label{game: effort}
	The effort exertion game is a tuple $\Gamma = (\mathcal{M,E},\boldsymbol{U})$ defined by:
	\begin{itemize}
		\item Players: The set $\mathcal{M}$ of members.
		\item Strategies: 
		Each member $m$ chooses an effort exertion strategy $e_m \in \mathcal{E}_m$. 
		The feasible set of all strategy profiles is $\mathcal{E} = \prod_{m\in \mathcal{M}} \mathcal{E}_m$.
		\item Payoffs: The vector $\boldsymbol{U} = (U_m,\ \forall m \in \mathcal{M})$ contains each member's payoff as defined in \eqref{payoff}.
	\end{itemize}
\end{game}

In Game \ref{game: effort}, given the other member's decision $e_{-m}$, member $m$ solves the following problem to find her best response.

\begin{prob}[Member $m$'s Best Response Problem] \label{pr: effort}
	\begin{equation} \label{br: effort}
		\begin{aligned}
			\max && U_m(e_m(\bm{d}),e_{-m}(\bm{d}))\\
			\text{s.t.} && e_m(\bm{d}) \in \mathcal{E}_m.
		\end{aligned}
	\end{equation}
\end{prob}

We aim to derive the Nash equilibrium as defined below. We use $\bm{s}=(s_m,\ \forall m \in \mathcal{M})$ to denote a generic strategy profile and $\mathcal{S}_m$ to denote a generic strategy space of member $m$ in Game \ref{game: incentive} and Game \ref{game: effort}.

\begin{defin}[Nash Equilibrium] \label{def: PNE}
	A strategy profile $\boldsymbol{s}^* = (s_m^*,s_{-m}^*)$ constitutes a Nash Equilibrium (NE), if for each $m \in \mathcal{M}$ and each $s_m' \in \mathcal{S}_m$, 
	\begin{align}
		U_m(s_m^*,s_{-m}^*) \geq U_m(s_m',s_{-m}^*),
	\end{align}	
	where $s_{-m}^*$ is the equilibrium strategy profile of the other member.
\end{defin}

Intuitively, an NE is a stable strategy profile where no member can increase her payoff by unilaterally changing her strategy. 

We list the key notations in Table \ref{tbl: notation} for ease of reading.

\begin{table}[t]
	\centering  % 显示位置为中间
	\caption{Key Notations.}  % 表格标题
	\label{tbl: notation}  % 用于索引表格的标签
	%字母的个数对应列数，|代表分割线
	% l代表左对齐，c代表居中，r代表右对齐
	\begin{tabular}{|c|c|}  
		\hline
		\multicolumn{2}{|c|}{\textbf{Variables}}\\
		\hline
		$d_m$ & Member $m$'s incentive contribution strategy\\ \hline
		$e_m$ & Member $m$'s effort exertion strategy\\ \hline
		$x_m$ & Member $m$'s reported solution\\ \hline
		$x_T$ & Aggregated team solution\\ \hline
		\multicolumn{2}{|c|}{\textbf{Parameters}}\\ \hline
		$\mathcal{M}$ & Set of members\\ \hline 
		$x$ & The underlying ground truth\\ \hline
		$V_H(V_L)$ & Member $H$ ($L$)'s valuation on team solution accuracy\\ \hline
		$D$ & Incentive volume\\ \hline
		$c$ & Effort exertion cost\\ \hline
		$a$ & Solution accuracy with effort exertion\\ \hline
		\multicolumn{2}{|c|}{\textbf{Functions}}\\ \hline
		$q_m(e_m)$ & The accuracy of member $m$'s solution\\ \hline
		$P_T(\bm{e})$ & The accuracy of aggregated team solution\\ \hline
		\multirow{2}{*}{$p_m^\text{z}(\boldsymbol{e})$} & Member $m$'s expected proportion from the incentive pool \\
		& under incentive allocation mechanism z\\ \hline
		$U_m(\boldsymbol{d,e})$ & Member $m$'s expected payoff\\ \hline
	\end{tabular}
	\vspace{-.5cm}
\end{table}

\subsection{Incentive Allocation Mechanisms} \label{subsec: alloc}

The equilibrium crucially depends on the incentive allocation mechanism. 
In this paper, we consider three incentive allocation mechanisms. The three mechanisms differ in the proportion of total incentives that each member can obtain from the incentive pool, and we use $p_m^\text{z}(\bm{e})$ to denote the proportion allocated to member $m$ under mechanism z.

% EA
\subsubsection{Equal Allocation (EA)}
EA is a benchmark mechanism where members equally split the incentive pool, i.e., 
\begin{equation} \label{eq: ea}
	p_m^{\text{EA}}(\bm{e}) = \frac{1}{|\mathcal{M}|},\ \forall m\in \mathcal{M}.
\end{equation}

% OA
\subsubsection{Output Agreement (OA)}
Output agreement \cite{liu2016learning} is a popular method in IEWV problems, where a member receives incentives if her reported solution is consistent with the majority. In this decentralized crowd decision setting, members split the pooled incentives if their reported solutions are consistent. Otherwise, they will not get any incentives.
The allocation under OA is shown in Lemma \ref{lem: oa}.
\begin{lem}[Allocated Proportion under OA] \label{lem: oa}
	The expected proportion of the total incentives that member $m$ can obtain from the incentive pool under OA is
	\begin{equation} \label{eq: oa}
		p_m^\text{OA}(\bm{e}) = \begin{cases}
			\frac{2a^2 -2a + 1}{2},\quad &\textup{if } e_H + e_L = 2,\\
			\frac{1}{4},& \textup{otherwise.}
		\end{cases}
	\end{equation}
\end{lem}

\begin{proof}[Proof sketch]
	We derive the expected proportion $p_m^{\text{OA}}(\bm{e})$ by first calculating the probability that one's reported solution matches the other's and then distributing the incentive pool evenly to the members with consistent solutions. 
\end{proof}

We provide the proof for Lemma \ref{lem: oa} in Section \uppercase\expandafter{\romannumeral2} of the online appendix \cite{online}. An essential implication of Lemma \ref{lem: oa} is that the expected proportion of the incentive amount can only be improved by both members' effort exertion. 
If only one member puts in the effort, the probability of two reported solutions matching is $1/2$ since the other member's report is random. Therefore, if at most one member exerts effort, they can only expect to receive a quarter of the incentive amount.

% SV
\subsubsection{Shapley Value (SV)}
The Shapley value is a solution concept in cooperative game theory \cite{shapley201617}, and it is widely applied in resource allocation and cost-sharing (e.g., \cite{an2019resource} \cite{oikonomakou2017fairness}). In our setting, a member's Shapley value $\phi_m(\bm{e})$ is defined as her average marginal contribution to the team solution accuracy, i.e., 
\begin{equation} \label{eq: sv}
	\begin{split}
		\phi_m(\boldsymbol{e}) = 
		\sum_{\mathcal{C}\subseteq \mathcal{M}\setminus\{m\}} \frac{ 1}{2}[p_T(\{e_n\}_{n\in \mathcal{C}\cup \{m\}}) - p_T(\{e_n\}_{n\in \mathcal{C}})].
	\end{split}
\end{equation}

Then, the proportion $p_m^{\text{SV}}(\bm{e})$ that member $m$ can get from incentive pool is
\begin{equation} \label{sv}
	p_m^{\text{SV}}(\bm{e}) = \frac{\phi_m(\boldsymbol{e})}{\sum_{n\in \mathcal{M}}\phi_n(\boldsymbol{e})}.
\end{equation}

\subsubsection{Mechanism Properties} \label{subsec: prop}
To facilitate a better understanding of the three mechanisms, we first define two properties in mechanism design and present the results in Table \ref{tbl: model summary}.

\begin{itemize}
	\item \emph{Individually Rational (IR):} An incentive allocation mechanism is IR if each member receives a non-negative payoff, i.e., $U_m(\bm{e},\bm{d}) \geq 0,\ \forall m\in \mathcal{M}.$
	\item \emph{Budget Balanced (BB)}: An incentive allocation mechanism is BB if the sum of each member's allocated incentives equals the amount of incentive contributions, i.e., $\sum_{m\in\mathcal{M}}p_m(\bm{e})\cdot \sum_{n \in \mathcal{M}} d_n = \sum_{m\in\mathcal{M}}d_m.$ Similarly, an incentive allocation mechanism is \textit{weakly Budget Balanced (weakly BB)} if $\sum_{m\in\mathcal{M}}p_m(\bm{e})\cdot \sum_{n \in \mathcal{M}} d_n \leq \sum_{m\in\mathcal{M}}d_m.$
\end{itemize}

\begin{table}[t]
	\centering  % 显示位置为中间
	\caption{Comparison in terms of Mechanism Properties.}  % 表格标题
	\label{tbl: model summary}  % 用于索引表格的标签
	%字母的个数对应列数，|代表分割线
	% l代表左对齐，c代表居中，r代表右对齐
	\begin{tabular}{|c|c|c|}  
		\hline
		\textbf{Incentive}& \textbf{Individual} & \textbf{Budget} \\
		\textbf{Allocation Mechanism} & \textbf{Rational (IR)} &\textbf{Balanced (BB)} \\
		\hline \textbf{Equal Allocation (EA)} & $\checkmark$ & $\checkmark$ \\	
		\hline \textbf{Output Agreement (OA)} & $\checkmark$ & Weakly \\ 		
		\hline \textbf{Shapley Value (SV)} & $\checkmark$ & $\checkmark$ \\
		\hline 
	\end{tabular}
\end{table}

We discuss the intuition of Table \ref{tbl: model summary} as follows.
\begin{itemize}
	\item All three mechanisms satisfy the IR property, since each member has the option to exert no effort and contribute no incentive, ensuring a non-negative payoff.
	\item Both EA and SV mechanisms are BB since the incentive pool is fully allocated to both members. Under OA, however, if members' reported solutions are inconsistent, they will receive nothing from the pool. Hence the OA mechanism is weakly BB. 
\end{itemize}

We discuss the social welfare under each mechanism in Section \ref{sec: comparison} and \ref{sec: simulation}. In the following, we derive the Nash equilibria under these three mechanisms and compare their mechanism performances.

\section{Equilibrium Analysis} \label{sec: NE}
In this section, we analyze members' equilibria under EA, OA, and SV, respectively. 
For ease of presentation, we will use a tuple $(\bm{d}^*=(d_H^*,d_L^*), \bm{e}^*=(e_H^*,e_L^*))$ to characterize the two-stage NE. 

In this paper, we focus on the case where $V_H/V_L >  \max\{\frac{14a-5}{6a-1},\ \frac{2a+1}{3-2a}\}$, i.e., members hold diverse valuations toward the team solution. For example, in a course project, a student values the team solution accuracy more since a high score is helpful for her postgraduate application. In contrast, the other student only needs a 'pass' and would rather spend more time on other issues (such as an industry internship).

For the completeness of our results, we have also analyzed the case where $1 < V_H/V_L\leq\max\{\frac{14a-5}{6a-1},\ \frac{2a+1}{3-2a}\}$. Interested readers can refer to online appendix \cite{online} for details.

\subsection{Equilibrium under Equal Allocation} \label{EA}
We solve members' equilibrium strategies under EA in Theorem \ref{thm: ea}.

\begin{thm}[Equilibrium under EA] \label{thm: ea}
	Define $c_L\triangleq \frac{2a-1}{4}V_L$ and $c_H\triangleq \frac{2a-1}{4}V_H$. Then
	\begin{equation}
		(\boldsymbol{d}^\ast, \boldsymbol{e}^*) =
		\begin{cases}
			((0,0),(1,1)), &\textup{ if } c\in (0,c_L],\ D\in (0,\infty),\\
			((0,0),(1,0)), &\textup{ if } c\in (c_L,c_H],\ D\in (0,\infty),\\
			((0,0),(0,0)), &\textup{ if } c\in (c_H,\infty),\ D\in (0,\infty).
		\end{cases}
	\end{equation}
\end{thm}

\begin{figure}[t]
	\centerline{\includegraphics[width=.95\linewidth]{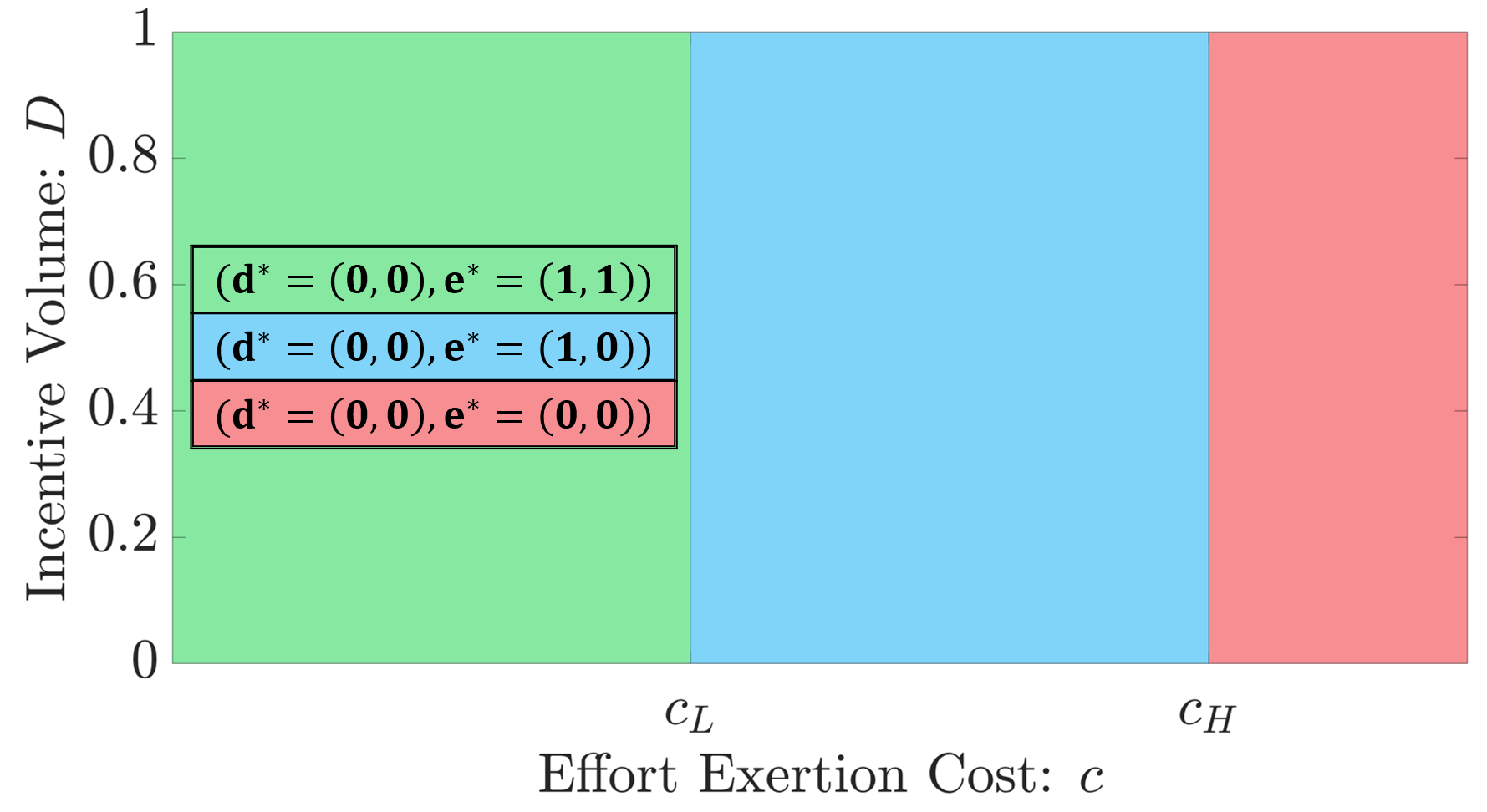}}
	\caption{Illustration of the NE under EA: The horizontal axis represents the effort exertion cost $c$ and the vertical axis represents the incentive contribution volume $D$. Regions of different colors represent different equilibrium profiles, and the legend shows the incentive contribution and effort exertion strategies at equilibrium. We illustrate the $(c, D)$ condition for the existence of each equilibrium profile.}
	\label{fig: ea}
\end{figure}

We provide the proof in Section \uppercase\expandafter{\romannumeral3} of the online appendix \cite{online}. An illustration of NE under EA is shown in Fig. \ref{fig: ea}. According to Theorem \ref{thm: ea}, we find that under all possible equilibrium profiles, neither member contributes any incentives, i.e., $\bm{d}^*=(0,0)$. Under EA, each member always wants to free-ride on the other one's  incentive contribution because she will always get half of the incentive pool regardless of her effort exertion strategy. This implies that EA fails to incentivize members to exert effort.

Furthermore, their effort exertion only depends on the effort cost, as shown in Fig. \ref{fig: ea}. Specifically, when the effort cost is low (\textit{the green part}), both members will exert effort to improve the team solution accuracy. When the cost is moderate (\textit{the blue part}), only member $H$ will exert effort, since she values the team solution accuracy more than member $L$. When the cost is high (\textit{the red part}), both of them will randomly report since the benefit from the improvement of team solution accuracy cannot compensate for the cost of effort exertion.

In summary, neither member will contribute incentives under EA to encourage members' effort exertion. To address this issue, we resort to the widely adopted mechanism in IEWV, i.e., Output Agreement (OA) (see \eqref{eq: oa}). We show the equilibrium results under OA in subsection \ref{subsec: oa}.

\subsection{Equilibrium under Output Agreement} \label{subsec: oa}

Output agreement scores members' solutions based on the solution similarity \cite{liu2016learning}. OA can motivate effort exertion since a member will get a larger reward from the incentive pool if she puts in effort (given the other one puts in effort as well) (see Lemma \ref{lem: oa}). We solve members' equilibrium strategies under OA in Theorem \ref{thm: oa}.

\begin{thm}[Equilibrium under OA] \label{thm: oa}
	Depending on the values of cost $c$ and incentive $D$, there exist four possible equilibria under OA, which are illustrated in Fig. \ref{fig: oa}.
\end{thm}

\begin{figure}[t]
	\centerline{\includegraphics[width=.95\linewidth]{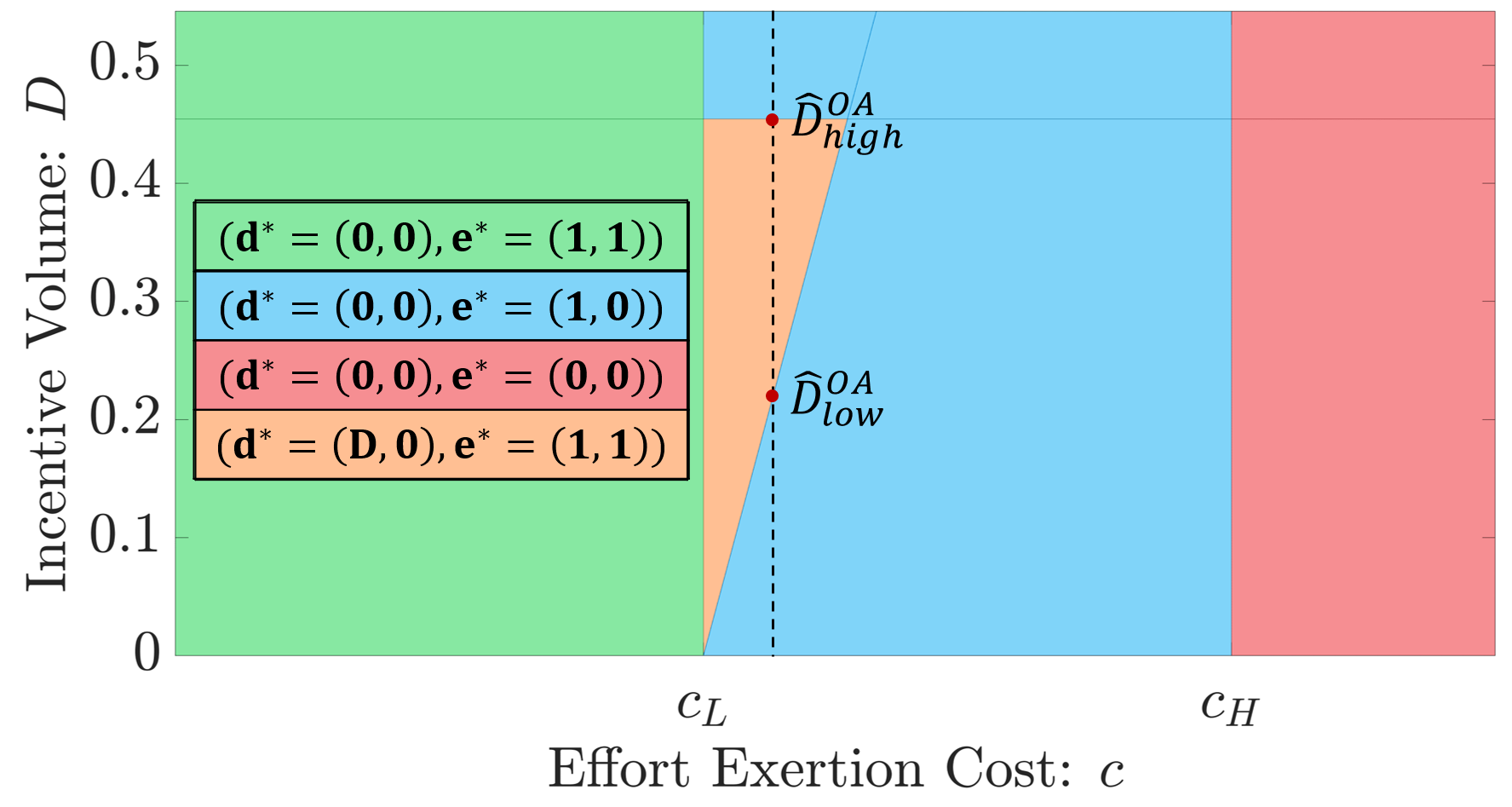}}
	\caption{Illustration of the NE under OA: We illustrate the $(c,D)$ condition for the existence of each equilibrium profile. $\hat{D}_\text{low}^\text{OA}(c)$ and $\hat{D}_\text{high}^\text{OA}(c)$ exist in the orange part and are two incentive volume thresholds for members' equilibrium strategies given an effort exertion cost $c$.}
	\label{fig: oa}
\end{figure}

One can find a detailed proof and the complete closed-form characterizations of the boundaries among various equilibria in Section \uppercase\expandafter{\romannumeral4} of the online appendix \cite{online}. Next, we discuss the intuitions behind Theorem \ref{thm: oa} using Fig. \ref{fig: oa}.

\paragraph{$(\bm{d}^* = (0,0), \bm{e}^* = (1,1))$ (the green part)}
When the effort exertion cost $c$ is lower than $c_L$, both members contribute none but exert effort. In this case, no incentives are needed as the utility of improving team solution accuracy exceeds the cost of effort. 

\paragraph{$(\bm{d}^* = (0,0), \bm{e}^* = (1,0))$ (the blue part)}
When the cost $c$ is moderate, i.e., $c_L < c \leq c_H$, there exist two cost-dependent incentive thresholds, $\hat{D}_\text{low}^\text{OA}(c)$ and $\hat{D}_\text{high}^\text{OA}(c)$. When the incentive volume $D$ is either below $\hat{D}_\text{low}^\text{OA}(c)$ or above $\hat{D}_\text{high}^\text{OA}(c)$, neither members contribute incentives, and only member $H$ exerts effort. Specifically, an incentive smaller than $\hat{D}_\text{low}^\text{OA}(c)$ cannot incentivize member $L$'s effort exertion, while an incentive larger than $\hat{D}_\text{high}^\text{OA}(c)$ hurts member $H$'s payoff. As a result, no member chooses to contribute incentives. However, member $H$ chooses to exert effort under a moderate cost as she has a higher valuation towards the team solution accuracy.

\paragraph{$(\bm{d}^* = (0,0), \bm{e}^* = (0,0))$ (the red part)}
When the cost $c$ is higher than $c_H$, neither members contribute incentives nor exert effort. With a high cost, one needs large incentives to motivate effort exertion. However, a large incentive volume will also hurt the contributor as she may not be able to get enough incentive allocation back to her. As a result, neither member contributes incentives nor exerts effort.  

\paragraph{$(\bm{d}^* = (D,0), \bm{e}^* = (1,1))$ (the orange part)}
When the cost $c$ is moderate, i.e., $c_L < c \leq c_H$ and the incentive volume $D$ is both above $\hat{D}_\text{low}^\text{OA}(c)$ and below $\hat{D}_\text{high}^\text{OA}(c)$, member $H$ is willing to provide a moderate amount of incentive to member $L$. This will motivate member $L$ to exert effort (together with member $H$), which will generate a team solution with higher accuracy.

Different from EA, OA can motivate members' effort exertion (see the orange part in Fig. \ref{fig: oa}, which does not exist under EA in Fig. \ref{fig: ea}). However, from Lemma \ref{lem: oa}, we find that under OA, a member cannot get a larger expected incentive proportion by exerting effort if the other one does not put in the effort. This property may not best incentivize effort exertion from members, which motivates us to use the Shapley value mechanism below.

\subsection{Equilibrium under Shapley Value} \label{subsec: sv}
In collaboration settings, it is common for a member to be able to observe each other's effort exertion, especially in groups with close interactions. For example, peer evaluation is widely adopted to assess the effort exertion in a project \cite{planas2021analysis}. With this information, we resort to the Shapley Value (SV) (see \eqref{sv}), which is built on a classical concept in cooperative game theory.

Members' equilibrium strategies under SV are shown in Theorem \ref{thm: sv}.

\begin{thm}[Equilibrium under SV] \label{thm: sv}
	Depending on the values of cost $c$ and incentive $D$, there exist five possible equilibria under SV, which are illustrated in Fig. \ref{fig: sv}.
\end{thm}

\begin{figure}[t]
	\centerline{\includegraphics[width=.95\linewidth]{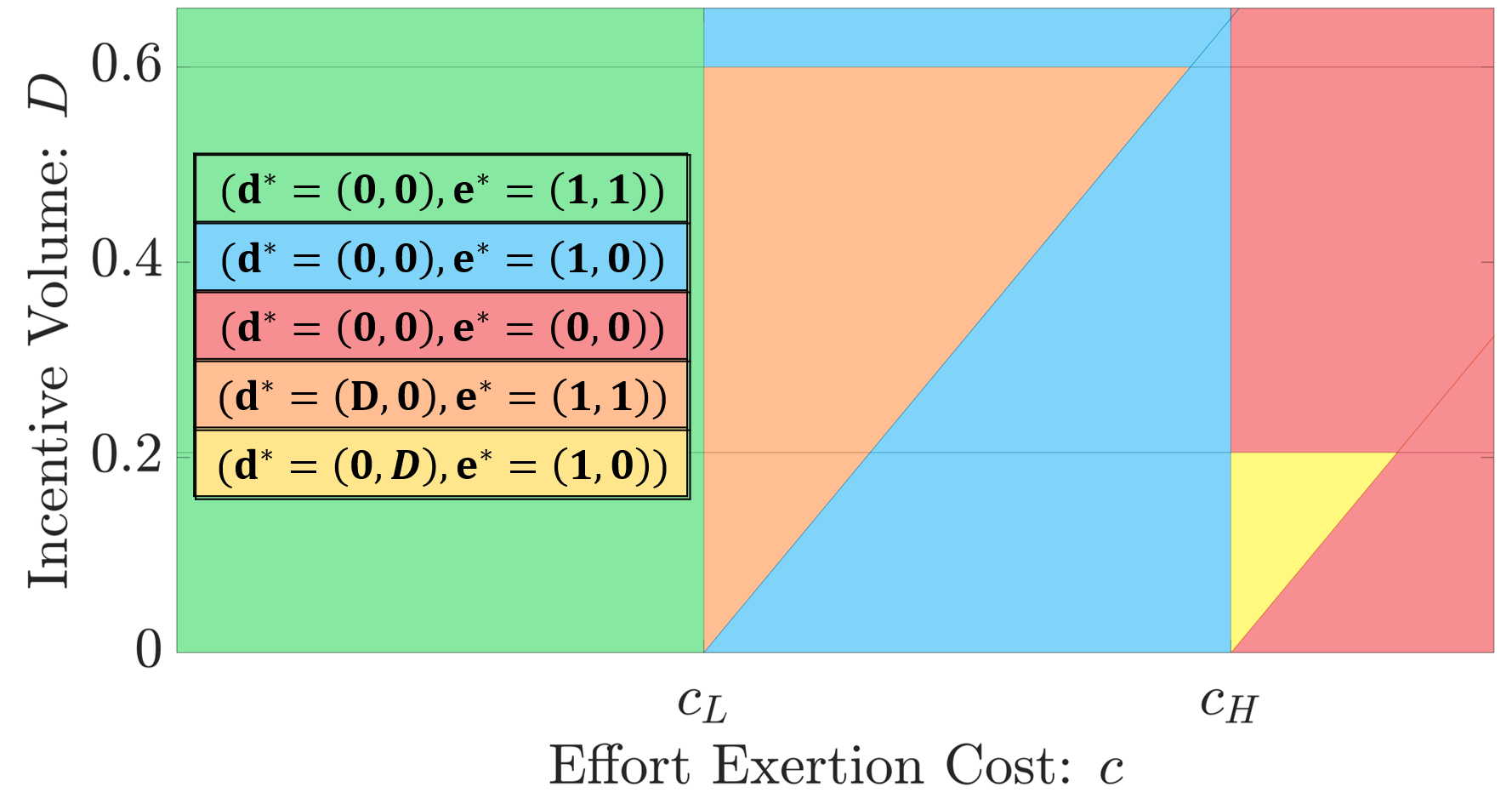}}
	\caption{Illustration of the NE under SV.}
	\label{fig: sv}
\end{figure}

We provide a detailed proof and the closed-form solutions to the equilibria in Section \uppercase\expandafter{\romannumeral5} of the online appendix \cite{online}. We find that four (out of five) equilibrium profiles under SV are identical to those under OA, and their intuitions are similar. Hence, we focus on the additional NE profile and explain its intuition as follows (illustrated in Fig. \ref{fig: sv}).

$\bullet$ $(\bm{d}^* = (0,D), \bm{e}^* = (1,0))$ \textit{(the yellow part):}
Surprisingly, there exists an NE where the low-valuation member incentivizes the high-valuation member to exert effort. 
When the cost $c$ is higher than $c_H$, but the incentive volume $D$ is small, member $L$ is reluctant to exert effort herself. However, she can contribute a small incentive so that member $H$ finds it beneficial to exert effort and generate a better team solution.

In all the mechanisms, we observe that whenever the low-valuation member exerts effort, the high-valuation member also exerts effort at the equilibrium. We formalize this result in Corollary \ref{non01}.

\begin{cor} \label{non01}
	At an equilibrium under any allocation mechanism, we have $e_H^\ast \geq e_L^\ast$. 
\end{cor}

Corollary 1 is proved by contradiction in Section \uppercase\expandafter{\romannumeral6} of the online appendix \cite{online}. The reason behind Corollary \ref{non01} is that as long as the low-valuation member can benefit from her own effort exertion, the high-valuation member can always improve her payoff by exerting effort since she values the team solution more. 

Next, we compare the equilibrium results under the three incentive allocation mechanisms.

\section{Mechanism Performance Comparison} \label{sec: comparison}
% intro
In this section, we compare the team performance at equilibrium under the three incentive allocation mechanisms in terms of 
\begin{itemize}
	\item Team solution accuracy $p_T^*$ (see \eqref{eq: team acc}).
	\item Social welfare $w^*=U_H^* + U_L^*$, i.e., the summation of members' payoffs.
\end{itemize}

% theorem: comparison
\begin{thm}[Team Equilibrium Performance Comparison] \label{thm: comp} 
	SV has the best team performance in terms of team solution accuracy and social welfare, i.e.,
	\begin{equation}\label{eq: comp}
		\begin{aligned}		
			& p_T^{\ast\text{SV}} \geq p_T^{*\text{OA}} \geq p_T^{*\text{EA}},\\
			& w^{*\text{SV}} \geq w^{*\text{OA}} \geq w^{*\text{EA}}.
		\end{aligned}
	\end{equation}
	The detailed comparison is illustrated in Fig. \ref{fig: comp}.	
\end{thm}

\begin{figure}[t]
	\centerline{\includegraphics[width=.95\linewidth]{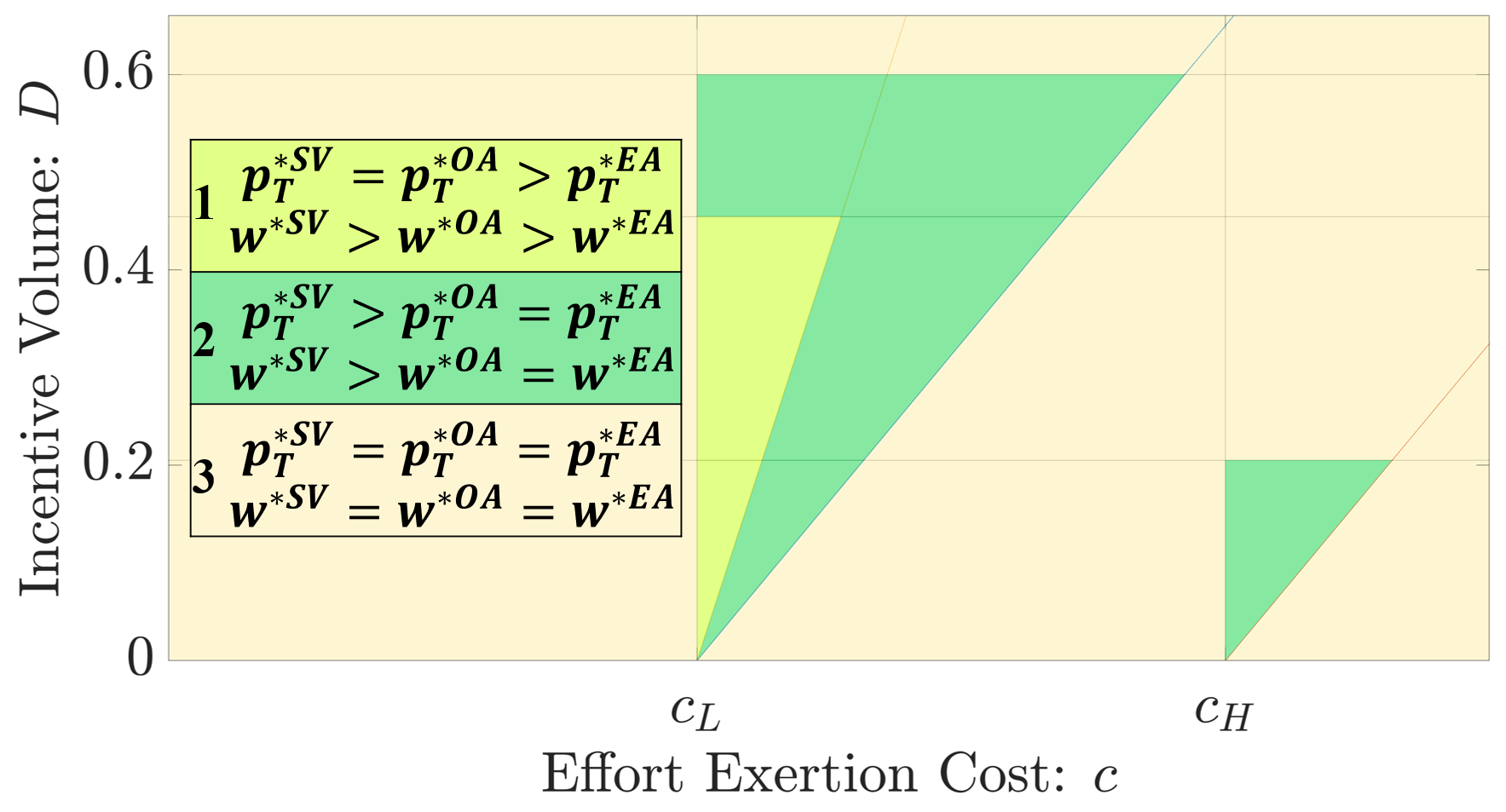}}
	\caption{Team equilibrium performance comparison at NE.}
	\label{fig: comp}
\end{figure}

One can find the closed-form comparison results in Section \uppercase\expandafter{\romannumeral7} of the online appendix \cite{online}. Next, we explain the intuition behind Theorem \ref{thm: comp}. As mentioned in Section \ref{EA}, EA fails to incentivize effort exertion since no one will contribute incentives. Hence, EA has the worst team solution accuracy and social welfare performance. 

Next, we discuss the team performance under OA and SV. First, recall the BB property in Table \ref{tbl: model summary}. Given the same effort exertion profile, the social welfare under SV is no smaller than that under OA (see Region 1 in Fig. \ref{fig: comp}) since the incentives can be fully allocated to the team members under SV. Second, a member under SV will always receive a non-negative payoff gain via her effort exertion. However, the consistency probability will not increase with one member's effort exertion alone under OA. Intuitively, given the effort exertion cost, it is easier to incentivize members under SV than OA to exert effort, which leads to better team performance (see Region 2 in Fig. \ref{fig: comp}).

To conclude, we have shown that SV performs the best regarding both team solution accuracy and social welfare.

\section{Numerical Results} \label{sec: simulation}

In this section, we conduct numerical experiments to validate our theoretical results and identify new insights. For the experiment setting, we set $a = 0.8,\ V_L = 1,$ and $V_H = 2$.

\begin{figure}[t]
	\centerline{\includegraphics[width=.9\linewidth]{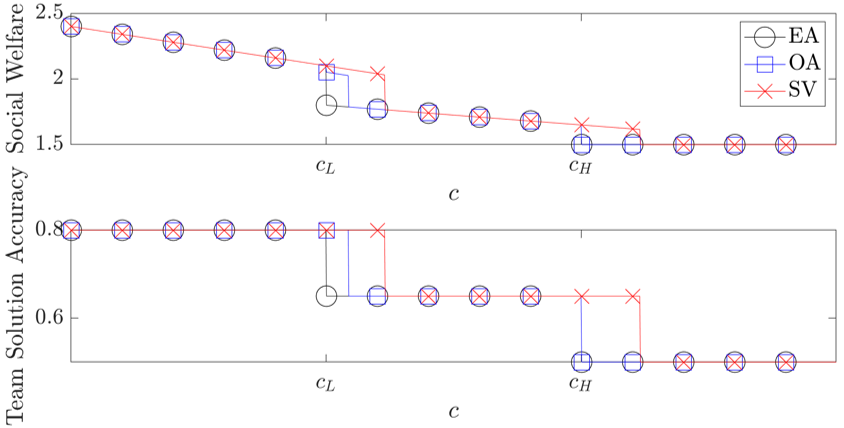}}
	\caption{Impact of $c$ on the team performance.}
	\label{fig: comp c}
\end{figure}

Fig. \ref{fig: comp c} shows how the team solution accuracy and social welfare depend on the effort exertion cost $c$ when $D=0.15$. From Fig. \ref{fig: comp c}, we observe:

\begin{observation}[Impact of Effort Exertion Cost on Team Equilibrium Performance]\label{ob: c decrease}
	Both the social welfare and the team solution accuracy are non-increasing in the effort exertion cost $c$ under three mechanisms.
\end{observation}

As the effort exertion cost increases, members become more reluctant to exert effort, making it more difficult to incentivize them. Consequently, there is a decrease in the team solution accuracy and social welfare.

\begin{figure}[t]
	\centerline{\includegraphics[width=.9\linewidth]{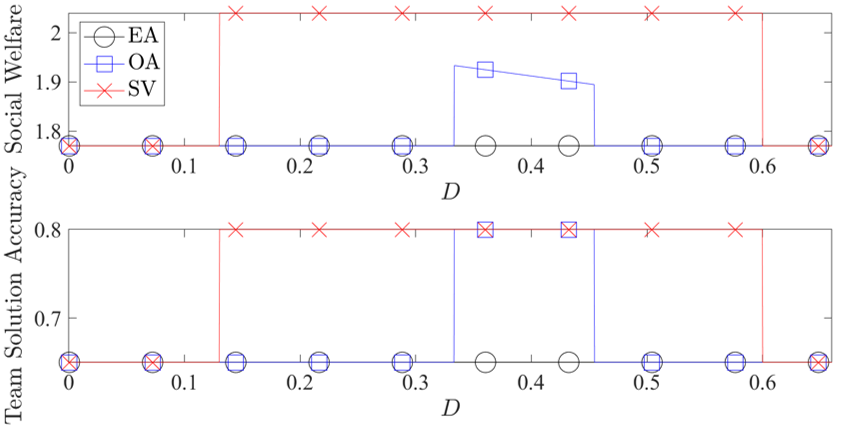}}
	\caption{Impact of $D$ on the team performance.}
	\label{fig: comp d}
\end{figure}

Fig. \ref{fig: comp d} shows how the team solution accuracy and social welfare depend on the incentive volume $D$ when $c=0.18$. From this figure, we observe:

\begin{observation}[Impact of Incentive Volume on Team Performance] \label{ob: d}
	Under SV and OA, both the social welfare and team solution accuracy first increase and then decrease in the incentive volume $D$.
\end{observation}

When the incentive volume is too small to incentivize effort exertion (e.g., $0 < D < 0.12$), no one will contribute incentives, and only member $H$ can benefit from exerting effort. When the incentive volume is enough to motivate member $L$'s effort exertion (e.g., $0.33 < D < 0.45$), member $H$ will contribute incentives, and both members will exert effort to improve the team performance. Finally, when the incentive volume is too large (e.g., $D > 0.6$), both members will choose not to contribute any incentives, and thus the team solution accuracy and social welfare decrease.

Fig. \ref{fig: comp c} and Fig. \ref{fig: comp d} also validate Theorem \ref{thm: comp} by showing that SV leads to team solution accuracy and social welfare no smaller than EA and OA.

\section{Conclusion} \label{sec: conclusion}
In this paper, we present the first theoretical study regarding incentive design in a decentralized IEWV problem where no central entity exists. We propose a decentralized incentive framework where members act as both incentive contributors and task solvers. We show the counter-intuitive result that there exists an equilibrium under SV where the low-valuation member contributes incentives while the high-valuation member does not. We further compare the three incentive allocation mechanisms via theoretical analysis and numerical experiments. We show that SV outperforms EA and OA in terms of both team solution accuracy and social welfare.

For future work, we plan to consider a more general multi-member scenario. Moreover, we will consider the interaction between members with multi-dimensional heterogeneity, such as accuracy levels and effort exertion costs. The analysis for bounded rational members in a decentralized IEWV scenario also deserves further research attention.

\bibliographystyle{IEEEtran}
\bibliography{ref}

\end{document}